\documentclass[compress,final,3p,times,11pt]{elsarticle}
\usepackage[utf8]{inputenc}
\usepackage{CJK}

\usepackage{amsfonts}
\usepackage{amssymb}
\usepackage{amsmath}
\usepackage{color,epstopdf}
\usepackage{mathrsfs}
\usepackage{graphicx,color,epstopdf}
\usepackage{float}
\usepackage{caption}
\usepackage{subcaption}
\usepackage{float}
\usepackage{subcaption}
\newtheorem{lemma}{Lemma}
\newtheorem{proof}{Proof}
\newtheorem{theorem}{Theorem}
\newtheorem{remark}{Remark}
\newtheorem{definition}{Definition}

 \journal{AIDE}
\begin{document}

\title{Dynamics of a Stochastic COVID-19 Epidemic Model with Jump-Diffusion \tnoteref{mytitlenote}}
\author[mymainaddress,mysecondaryaddress]{Almaz Tesfay\corref{mycorrespondingauthor}}
\ead{amutesfay@hust.edu.cn}
\author[myfourthaddress]{Tareq Saeed}
\ead{tsalmalki@kau.edu.sa}
\author[mythirdaddress]{Anwar Zeb}
\ead{anwar@cuiatd.edu.pk}
\author[mymainaddress,mysecondaryaddress]{Daniel Tesfay}
\ead{dannytesfay@hust.edu.cn}
\author[mymainaddress]{Anas Khalaf}
\ead{anasdheyab@hust.edu.cn}
\author[myfifthdaddress]{James Brannan}
\ead{jrbrn@clemson.edu}
\cortext[mycorrespondingauthor]{Corresponding author}
\address[mymainaddress]{School of Mathematics and Statistics \& Center for Mathematical Sciences,  Huazhong University of Science and Technology, Wuhan 430074, China}
\address[mysecondaryaddress]{Department of Mathematics, Mekelle University, P.O.Box 231, Mekelle, Ethiopia}
\address[myfourthaddress]{Department of Mathematics,
King Abdulaziz University, Jeddah, 41206, Kingdom of Saudi Arabia}
\address[mythirdaddress]{Department of Mathematics, COMSATS University Islamabad, Abbottabad Campus, Abbottabad, 22060, Khyber Pakhtunkhwa, Pakistan}
\address[myfifthdaddress]{Department of Mathematical Sciences, Clemson University,Clemson, South Carolina 29634, USA}
\begin{abstract}
For a stochastic COVID-19 model with jump-diffusion, we prove the existence and uniqueness of the global positive solution. We also investigate some conditions for the extinction and persistence of the disease.  We calculate the threshold of the stochastic epidemic system which determines the extinction or permanence of the disease  at different intensities of the stochastic noises. This threshold is denoted by $\xi$ {which} depends on the white and jump noises. {The effects of these noises on the dynamics of the model are studied.} {The numerical experiments
show that the random perturbation introduced in the stochastic model suppresses disease outbreak as compared to its deterministic counterpart. In other words, the impact of the noises on the extinction and persistence is high.} When the noise is large or small, our numerical findings show that the COVID-19 vanishes from the {population} if $\xi <1;$ whereas {the epidemic can't go out of control} if $\xi >1.$ From this, we observe that white noise and jump noise have a significant effect on the spread of COVID-19 infection, \textcolor{black}{i.e., we can conclude that the stochastic model is more realistic than the deterministic one. Finally,}
to illustrate this phenomenon, we put some numerical simulations.
\end{abstract}
\maketitle

\section{Introduction}

Infectious diseases are the public enemy of \textcolor{black}{ the} human population and have brought \textcolor{black}{a} great impact \textcolor{black}{ on} mankind. In the present time, the novel coronavirus is the major disease in the world. This new strain of coronavirus is called COVID-19 or SARS-Cov2.  COVID-19 has been declablack \textcolor{black}{by the World Health Organization} as a global emergency  \textcolor{black}{in} \textcolor{black}{ January 2020}, and a pandemic on {March 2020}  \cite{SOHRABI202071}. Since the first breakout of the pandemic, according to the data released by World barometer \cite{WHO2020Nov}, there are more than 52 million confirmed (from which 17 million are active) cases, 1.29 million deaths and 33.5 million recoveries from the disease.

 Researchers are working around the clock to understand the nature of the disease deeply. Scientists are also battling to produce a vaccine {for} this new virus.

 Numerous scholars have conducted investigation to pblackict the spread of the COVID-19 in order to seek the best prevention measures. For example \cite{2020A,hou2020effectiveness,kucharski2020early,okuonghae2020analysis,aa,yang2020mathematical} studied mathematical models of COVID-19 to describe the spread of the coronavirus. {Stochastic transition models  were established in \cite{chen2020mathematical,wang2020phase,anas} to evaluate the spread of COVID-19. The importance of isolation and quarantine was also  emphasized in those articles. Dalal et al. \cite{dalal2008stochastic}  studied the impact of the environment in the AIDS model  using the method of parameter perturbation.} In papers \cite{at1,at2,at3,at4,d2,d3,d4,d5} fractal-fractional differentiation and integration is discussed. This approach is very important in investigating the stochastic COVID-19 model.

Stochastic dynamical systems are widely used to describe different complex phenomena. The random fluctuations in complex phenomena usually portray intermittent jumps i,e., the noises are non-Gaussian.  In other words, epidemic models are inevitably subject to environmental noise and
it is necessary to reveal how the environmental noise influences the epidemic model. { In the natural world, there are different types of random noises, such as the well known white noise, the L\'evy jump noise which considers the motivation that the continuity of solutions may be inevitably} under severe environmental perturbations, such as earthquakes, floods, volcanic eruptions, SARS, influenza  \cite{bao2012stochastic,liu2018analysis,zhang2017dynamics} and a jump process should be introduced to prevent and control diseases, and so on. Mathematically, {several authors} \cite{sun2020dynamics,sun2020dy,danny,danny1} {used}
the L\'evy process to describe the phenomena that cause a big jump to occur occasionally.

Recently, {Zhang et al}. \cite{zhang2020dynamics} investigated the  stochastic COVID-19 mathematical model driven by Gaussian noise. The authors assumed that environmental fluctuations in the constant $\beta$, so that $\beta \longrightarrow \beta +\lambda \dot {B}_t$ where $B_t$ is a one dimensional Brownian motion \cite{tesfay2021logistic}. The stochastic COVID-19 model which they consideblack is
\begin{align}\label{CV-BM}
 &dS_t= (\Lambda-\beta\,S_t\,I_t-\nu\,S_t+\sigma \,R_t)dt-\lambda\,S_t\,I_t\,dB_t, \nonumber\\
 &dI_t=( \beta\,S_t\,I_t-(\nu+\gamma)\,I_t)dt+\lambda\,S_t\,I_t\,dB_t,\nonumber\\
 &dR_t=(\gamma\,I_t-(\nu+\sigma)\,R_t)dt,
\end{align}
where the variables $S_t,\, I_t$, and $R_t$ represent the susceptible population, infectious population, and recoveblack (removed) population, respectively. The parameters $\Lambda,\, \beta,\, \nu,\, \gamma$ and  $\sigma$ are all positive constant numbers, and they represent the joining rate of {the} population to susceptible class through birth or migration, {the} rate at which the susceptible tend to infected class (like social distancing $\beta\in(0,1)$), due to natural cause and from COVID-19, the recovery rate, and the rate of health deterioration, respectively. $B_t$ is standard Brownian motion defined on the complete probability space $(\Omega, \mathcal{F},\{\mathcal{F}_t\}_{t\geq 0}, \mathbb{P})$ and $\lambda$ is {the} intensity of the Gaussian noise. The { researchers} proved the existence and uniqueness of the non-negative solution of the system (\ref{CV-BM}), and they also showed the extinction and persistence of the disease. But they did not consider the jump noise.

Since, the stochastic model (\ref{CV-BM}) that does not take randomness can not efficiently model these phenomena. The L\'evy noise, which is more comprehensive, is a better candidate \cite{applebaum2009levy}.

Here, we consider that the environmental Gaussian and non-Gaussian noises are directly proportional to the state variables $S_t, \, I_t,$ and $R_t.$ Several scholars used this approach, for instance, we refer to { \cite{berrhazi2018stochastic,kiouach2020long,zhou2016threshold}} and references therein. The system which we consider has the following form:
\begin{align}\label{CV-BM-LM}
 &dS_t= (\Lambda-\beta\,S_{t-}\,I_{t-}-\nu\,S_{t-}+\sigma \,R_{t-})\,dt+\lambda_1\,S_{t-}\,dB^1_t +\int_{\mathbb{Y}} \epsilon_{1}(y)S_{t-}\,\bar{N}(dt,dy),\nonumber\\
 &dI_t=( \beta\,S_{t-}\,I_{t-}-(\nu+\gamma)\,I_{t-})\,dt+\lambda_2\,I_{t-}\,dB^2_t+\int_{\mathbb{Y}} \epsilon_{2}(y)I_{t-}\,\bar{N}(dt,dy),\nonumber\\
 &dR_t=(\gamma\,I_{t-}-(\nu+\sigma)\,R_{t-})\,dt+\lambda_3\,R_{t-}\,dB^3_t+\int_{\mathbb{Y}} \epsilon_{3}(y)R_{t-}\,\bar{N}(dt,dy),
\end{align}
where $S_{t-}$ is the left limit of $S_t$. The description of the parameters $\Lambda,\, \beta,\, \nu,\, \gamma$ and  $\sigma$  are the same as in {the} model (\ref{CV-BM}). For $j=1,2,3$,  $\epsilon_{j}(y)$ is a bounded function satisfying $\epsilon_{j}(y)+1>0$ on the intervals $|y|\geq 1$ or $|y|<1$. $N(t,dy)$ is the independent Poisson random measure on $\mathbb{R}^+\times{{\mathbb{R}}\setminus{\{0\}}}$, $\bar{N}(t,dy)$ is the compensated Poisson random  measure satisfying $\bar{N}(t,dy)=N(t,dy)-\pi(dy)dt$, where $\pi(.)$ is a $\delta$-finite measure on a measurable subset $\mathbb{Y}$ of $(0,\infty)$ and $\pi(\mathbb{Y})<\infty$,  \cite{applebaum2009levy,tesfay2020mean}. $B^j_t$ are mutually independent standard Brownian motion and $\lambda_j$ stand for the intensities of the Gaussian noise,  \cite{duan2015introduction}. To the best of our {knowledge,} this model is not studied before.

In this study, we are going to investigate the stochastic COVID-19 model  with jump-diffusion (\ref{CV-BM-LM}). { The existence of the  solution of the stochastic model (\ref{CV-BM-LM})
is analyzed. We use the Euler Maruyama (EM) method, which is proposed in \cite{higham2001algorithmic,1992Higher}, after revising and changing it a bit to fit our model. The consistency, convergence, and stability of this numerical method is also proved in the afore-mentioned papers. This method helps to evaluating explanations based on the notion of adversarial robustness.} { Using numerical simulations, we study the impact of the deterministic parameters and noise intensities on the proposed system. We think this is a better tool to demonstrate the interactions between the epidemic system and its complex surrounding. Especially, we
focus on the extinction and persistence of the  SARS-Cov2 and present the biological interpretations. The evaluation criteria further allows us to derive new explanations which capture pertinent features qualitatively and quantitatively. From the plotted figures, we can observe that the noise intensities have a great impact on the systems (\ref{CV-Det}) and (\ref{CV-BM-LM}). More details are given in Sections \ref{Sec3} and \ref{s6}.}

The goal of the present work is to make contributions to understand the dynamics of the novel disease (COVID-19) epidemic models with both Gaussian and non-Gaussian noises, {i.e., we aspire to study the effect of Gaussian noise and jumps intensities on COVID-19 epidemic.}

The rest of the paper is constituted as follows. In Section \ref{s2}, we recall some important notations and lemmas. In section \ref{determ}, we {discuss} the dynamical behaviour of the deterministic COVID-19 model. Section \ref{Sec3} has two subsections. The existence and uniqueness of the solution of the stochastic COVID-19 model (\ref{CV-BM-LM}) is given in subsection \ref{s3}. While in Subsection \ref{s4}, by finding the value of the threshold, we show the conditions for the extinction and persistence to COVID-19. The discussion and numerical experiments of our work are given in Section \ref{s6}. {Finally, we put conclusion of our study in Section \ref{con}.}

\section{Preliminaries}\label{s2}
In this section, we will introduce some basic notations and lemmas. Throughout this paper, we have

\begin{description}
  \item[a.] $(\Omega, \mathcal{F},\{\mathcal{F}_t\}_{t\geq 0}, \mathbb{P})$ denotes a complete  filteblack probability space;
  \item[b.] $\mathbb{R}^3_{+}:=\{x=(x_1,x_2,x_3)\in \mathbb{R}^3: x_j\geq 0, \, j=1,2,3\}$, \,\, $\mathbb{R}_+=(0,\infty)$;
  \item[c.] For the jump-diffusion, let $n\geq 0$, there is a positive constant $L_n$ such that

  (i.) $\int_{\mathbb{Y}} |H_{j}(x,y)-H_{j}(\bar{x},y)|^2\,\pi (dy) \leq L_n\,|x-\bar{x}|^2$ where $H_{j}(x,y)=\epsilon_{j}(y)\,X_t,\, \, j=1,2,3$. For more details we refer to \cite[P. 78]{Siakalli2009Stability}, \cite{0Impulsive};\\
 (ii.) $1+\epsilon_{j}(y)\geq 0, \, y\in \mathbb{Y}, j=1,2,3$, there exists $C > 0$ such that $\int_{\mathbb{Y}} (\ln(1+\epsilon_{j}(y)))^2\,\pi(dy) < C$;
  \item[d.] $<M>_{t}=\frac{1}{t}\int_{0}^{t}M_rdr$,\,\,\, $<M>^{*}_{t}=lim _{t\rightarrow \infty}inf \frac{1}{t}\int_{0}^{t}M_rdr$,\,\,\, $<M>_{t}^{**}=lim _{t\rightarrow \infty}sup \frac{1}{t}\int_{0}^{t}M_rdr$ ;
   \item[e.] For $j=1,2,3$, $\varphi_{j}=\frac{\lambda^2_{j}}{2}+\int_{\mathbb{Y}} (\epsilon_{j}(y)-\ln(1+\epsilon_{j}(y)))\,\pi(dy), \, j=1,2,3$;
   \item[f.]\label{f} $\psi_{j}=\int_{\mathbb{Y}} (\ln(1+\epsilon_{j}(y)))\,\bar{N}(dt,dy)$,\quad $ <\psi_j,\psi_j>=t\int_{\mathbb{Y}} (\ln(1+\epsilon_{j}(y)))\,\pi(dy)<t\,C$;

   \item[g.] For some positive $m>2$, \,\ $M=\nu-\frac{m-1}{2}\,\bar{\Lambda}^2-\frac{1}{m}\,\bar{\epsilon}$, \, where $\bar{\Lambda}=max\{\lambda_1^2, \lambda_2^2, \lambda_3^2\}$, and $\bar{\epsilon}=\int_{\mathbb{Y}}(1+\tilde{\epsilon})^m-1-m\,\hat{\epsilon}\,\,\pi(dy)$, where $\tilde{\epsilon}=max\{\epsilon_1(y), \epsilon_2(y), \epsilon_3(y)\}$, and $\hat{\epsilon}=min\{\epsilon_1(y),\epsilon_2(y),\epsilon_3(y)\}$
   \item[h.] $inf \emptyset=\infty$\, where $\emptyset$ denotes empty set.
\end{description}
\begin{remark}\label{rem-1}
For some positive $x$, the following is true, $x-1-\ln x >0.$
\end{remark}
\begin{lemma}\label{Ito-form}
(The one dimensional It$\hat{o}$ formula). Here we will give It$\hat{o}$ formula for the following $n$-dimensional stochastic differential  equation (SDE) with jump noise \cite{applebaum2009levy}
\begin{align}\label{SDE-Ito}
dY(t)=G(Y(t))dt+F(Y(t))dB_t+\int_{|y|<1}H(Y(t),y)\bar{N}(dt,dy) \quad t\geq 0,
\end{align}
where $G:\mathbb{R}_+\times \mathbb{R}^n\rightarrow\mathbb{R}^n$,\quad $F:\mathbb{R}_+\times \mathbb{R}^n\rightarrow\mathbb{R}^n\times\mathbb{R}^d$,\quad $H:\mathbb{R}_+\times \mathbb{R}^n\times\mathbb{R}^n\rightarrow\mathbb{R}^n$, \quad for $n\geq 2$\quad are consideblack as measurable.

Assume $Y$ be a solution of the SDE (\ref{SDE-Ito}). Then, for each $W \in C^2(\mathbb{R}^n)$, $t\in[0,\infty)$, with probability one,  we have \cite{Siakalli2009Stability}
\begin{align*}
W(Y(t))-W(Y(0))&=\int_{0}^{t} \partial_{j}W(Y_{c}(r^{-}))dY^j +\frac{1}{2}\int_{0}^{t}\partial_{j}\partial_{i}W(Y_{c}(r^{-}))d[Y^j_{c},Y^i_{c}](r)\nonumber\\
&+\int_{0}^{t}\int_{|y|<1}[W(Y(r^-)+H(Y(r),y))-W(Y(r^{-}))]\bar{N}(dr,dy)\nonumber\\
&+\int_{0}^{t}\int_{|y|<1}[W(Y(r^{-})+H(Y(r),y))-W(Y(r^{-}))-H^i(Y(r),y)\,\partial_{i}W(Y(r^{-}))]\pi(dy)dr,
\end{align*}
\end{lemma}
where $Y_c$ is the continuous part of $Y$ given by $Y^i_c(t)=\int_0^tF^i_{k}(s)dB^k(s)+\int ^t_{0}G^i(s)ds$,\quad $1\leq i\leq n,\,1\geq k\leq m, \, t\geq 0.$
The proof of this lemma is given in \cite[P. 226]{applebaum2009levy}.

Next, let us denote $LW:[0,\infty)\times \mathbb{R}^n\rightarrow \mathbb{R}$ as the linear function associated to the SDE (\ref{SDE-Ito}) which is given by
\begin{align*}
(LW)(\eta)&= G^{i}(\eta)(\partial_{i}W)(\eta(0))+\frac{1}{2}[F(\eta)(F(\eta)^T]^{ik}(\partial_{i}\partial_{k}W)(\eta(0))\nonumber\\   &+\int_{|y|<1}[W(\eta(0)+H(\eta,y))-W(\eta(0))-H^{i}(\eta,y)(\partial_{i}W)(\eta(0))]\pi(dy),
\end{align*}
where $\eta\in[0,\infty)\times \mathbb{R}^n .$

\begin{lemma}\label{lemma-1}
Assume $(c)$ holds. The stochastic model (\ref{CV-BM-LM}) has a unique non-negative solution $(S_t,I_t,R_t)\in \mathbb{R}_+^3$ for any given initial value $(S_0,I_0,R_0)\in \mathbb{R}^3_+$ on time $t\geq 0$ almost surely (a.s.). Under $(g)$, the solution of model (\ref{CV-BM-LM}) satisfies the following conditions:\\

(i.) $lim_{t \rightarrow \infty} \left(\frac{S_t+I_t+R_t}{t}\right)=0$ a.s.\\
Moreover, $lim_{t \rightarrow \infty} \left(\frac{S_t}{t}\right)=0, \,\, lim_{t \rightarrow \infty} \left(\frac{I_t}{t}\right)=0, \,\,lim_{t \rightarrow \infty} \left(\frac{R_t}{t}\right)=0,$\\
(ii.) $lim_{t \rightarrow \infty}\frac{S_tdB^1_t}{t}=0$,\,\,  $lim_{t \rightarrow \infty}\frac{I_tdB^2_t}{t}=0$, \,\, $lim_{t \rightarrow \infty}\frac{R_tdB^3_t}{t}=0$,
 $lim_{t \rightarrow \infty}\frac{\int_0^t \int_{\mathbb{Y}}S_r\,\epsilon_{1}(y)\,\bar{N}(dr,dy)}{t}=0$,\\

 $lim_{t \rightarrow \infty}\frac{\int_0^t \int_{\mathbb{Y}}I_r\,\epsilon_{2}(y)\,\bar{N}(dr,dy)}{t}=0$,\,\,$lim_{t \rightarrow \infty}\frac{\int_0^t \int_{\mathbb{Y}}R_r\,\epsilon_{3}(y)\,\bar{N}(dr,dy)}{t}=0.$ \quad a.s.
\end{lemma}
\begin{proof}
The proof of this lemma is similar to \cite{zhou2016threshold} and hence is omitted.
\end{proof}
\section{Dynamical analysis of deterministic COVID-19 model}\label{determ}

The deterministic version of systems (\ref{CV-BM}) and (\ref{CV-BM-LM}) is

\begin{align}\label{CV-Det}
 &\frac{dS_t}{dt}= \Lambda-\beta\,S_t\,I_t-\nu\,S_t+\sigma \,R_t,\nonumber\\
 &\frac{dI_t}{dt}=\beta\,S_t\,I_t-(\nu+\gamma)\,I_t,\nonumber\\
 &\frac{dR_t}{dt}t=\gamma\,I_t-(\nu+\sigma)\,R_t,
\end{align}
and
\begin{align}\label{CD_X}
  \frac{dX}{dt}=\frac{dS_t}{dt}+\frac{dI_t}{dt}+\frac{dR_t}{dt}= \Lambda-\nu\,X,
\end{align}
where $X=S_t+I_t+R_t$. For $\Lambda=\nu\,X,$ Equation (\ref{CD_X}) shows $X$ is the total constant population with initial value $X_0=S_0+I_0+R_0$. This equation has analytical solution
\begin{align}\label{CD_X_sol}
   X=\frac{\Lambda}{\nu}+X_0\,e^{-\nu\,t}.
\end{align}
Since, the initial values are non-negative, we have $S_t\geq 0, \, I_t\geq 0,\, R_t\geq 0,$ and $lim_{t\rightarrow \infty}X=\frac{\Lambda}{\nu}.$ One can easily conclude that   $0 < X \leq \frac{\Lambda}{\nu}.$ Therefore, Eq. (\ref{CD_X_sol}) has a positivity property. Thus the deterministic COVID-19 model (\ref{CV-Det}) is biologically meaningful and bounded in the domain
$$\mathbb{D}=\left\{(S_t,I_t,R_t)\in \mathbb{R}^{3}_{+}: 0 < X \leq \frac{\Lambda}{\nu}\right\}{.}$$

The equilibrium point of system (\ref{CV-Det}) satisfies the following:
\begin{align*}
  & \Lambda-\beta\,S_t\,I_t-\nu\,S_t+\sigma \,R_t=0,\nonumber\\
 &\beta\,S_t\,I_t-(\nu+\gamma)\,I_t=0,\nonumber\\
 &\gamma\,I_t-(\nu+\sigma)\,R_t=0,
\end{align*}
having the equilibria:
\begin{align*}
    &E^0=(S^0,I^0,R^0)=\left(\frac{\Lambda}{\nu},0,0\right){,}\nonumber\\
    &E^1=(S^1,I^1,R^1)=\left(\frac{\nu+\gamma}{\beta},\,\frac{\beta\,\Lambda-\nu\,(\nu+\gamma)}{\nu+\gamma},0\right){,}\nonumber\\
    & E^2=(S^2,I^2,R^2)=\left(\frac{\nu+\gamma}{\beta},\,0,\frac{\nu\,(\nu+\gamma)-\beta\,\Lambda}{\beta\,\sigma}\right){,}\nonumber\\
    & E^3=(S^3,I^3,R^3)=\left(\frac{\nu+\gamma}{\beta},\,\frac{(\Lambda-\nu\,S^3)(\nu+\sigma)}{\beta\,S^3(\nu+\sigma)-\gamma\,\sigma},\,\frac{\gamma\,(\Lambda-\nu\,S^3)}{\beta\,S^3(\nu+\sigma)-\gamma\,\sigma} \right){,}
\end{align*}
where $S^3=\frac{\nu+\gamma}{\beta}$.\\
$E^0$ is called disease-free equilibrium point (free virus equilibrium point). Because there are no infectious individuals in the population, which indicates that $I = 0$ and $R = 0$. $E^{3}$ is known as endemic equilibrium point (the positive virus point ) of the model (\ref{CV-Det}).

From the expressions of $I^1$ and $I^3$, noting that if
$$\frac{\Lambda}{\nu}>\frac{\nu+\gamma}{\beta}{,}$$
the deterministic system (\ref{CV-Det}) has unique positive equilibrium $E^1$ and $E^3$. From this the reproductive number of the system (\ref{CV-Det}) is given by
\begin{align*}
    \xi_0= \frac{\beta\,\Lambda}{(\nu+\gamma)\,\nu}{.}
\end{align*}

Similarly, at equilibrium point $E^3$, all the eigenvalues are non-positive if $\xi_0 > 1$. Hence the proposed model
is globally stable if  $\xi_0 > 1$.
\begin{theorem}\label{Det-thrm}
The deterministic system (\ref{CV-Det}) has\\
(i) a unique stable ‘disease-extinction’ (disease-free equilibrium) equilibrium point $E^j$ for $j=0,1,2,3$ if $\xi_0 <1$. This
indicates the extinction of the disease from the population..\\
(ii) a stable positive equilibrium $E^j$ for $j=0,1,2,3$ exists if $\xi_0 >1$ that shows the permanence of the
disease.
\end{theorem}
\begin{proof}
The Jacobian matrix of the system (\ref{CV-Det}) is

$$J=\left(
  \begin{array}{ccc}
    -\beta\,I-\nu & -\beta\,S & \sigma \\
    \beta\,I & \beta\,S-(\nu+\gamma) & 0 \\
   0 & \gamma & -(\nu+\sigma) \\
  \end{array}
\right){.}$$
Now let us show for $j=0$ ($E^0$), then similarly can show for $j=1,2,3$.

The Jacobian of the system (\ref{CV-Det}) at $E^0$ obtains\\

$$J^0=\left(
  \begin{array}{ccc}
    -\nu & -\beta\,\frac{\Lambda}{\nu} & \sigma \\
   0 & \beta\,\frac{\Lambda}{\nu}-(\nu+\gamma) & 0 \\
   0 & \gamma & -(\nu+\sigma) \\
  \end{array}
\right){.}$$

The eigenvalues are calculated as follows:\\
\begin{align}\label{E-0}
  J^{E^0}=\left|
  \begin{array}{ccc}
    -\nu-\bar{\lambda} & -\beta\,\frac{\Lambda}{\nu} & \sigma \\
   0 & \beta\,\frac{\Lambda}{\nu}-(\nu+\gamma)-\bar{\lambda} & 0 \\
   0 & \gamma & -(\nu+\sigma) -\bar{\lambda}\\
  \end{array}
\right| {.}
\end{align}
The characteristic polynomial of equation (\ref{E-0}) is

$$( -\nu-\bar{\lambda})(\beta\,\frac{\Lambda}{\nu}-(\nu+\gamma)-\bar{\lambda})(-(\nu+\sigma) -\bar{\lambda})=0,$$
so the eigenvalue is

$$\bar{\lambda}=\beta\,\frac{\Lambda}{\nu}-(\nu+\gamma){.}$$

From the stability theory, $E^0$ is stable if and only if\\

$\bar{\lambda}<0,$\\

or equivalently\\

$$\beta\,\frac{\Lambda}{\nu}-(\nu+\gamma)<0,$$ implies $$\xi_0=\beta\,\frac{\Lambda}{\nu\,(\nu+\gamma)}<1.\, \quad \Box$$
\end{proof}

\begin{figure}[htb!]
\centering
  \subfloat[Solution of $\frac{dI_t}{dt}$ ]{\includegraphics[width=0.5\textwidth]{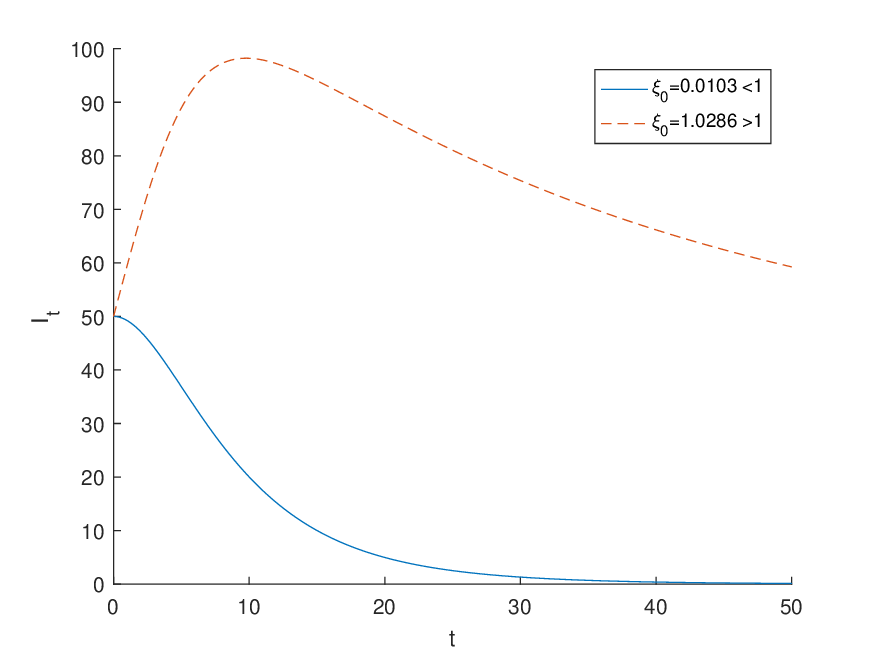}}
\subfloat[$\frac{dI_t}{dt}$ versus $\nu$ .]{\includegraphics[width=0.5\textwidth]{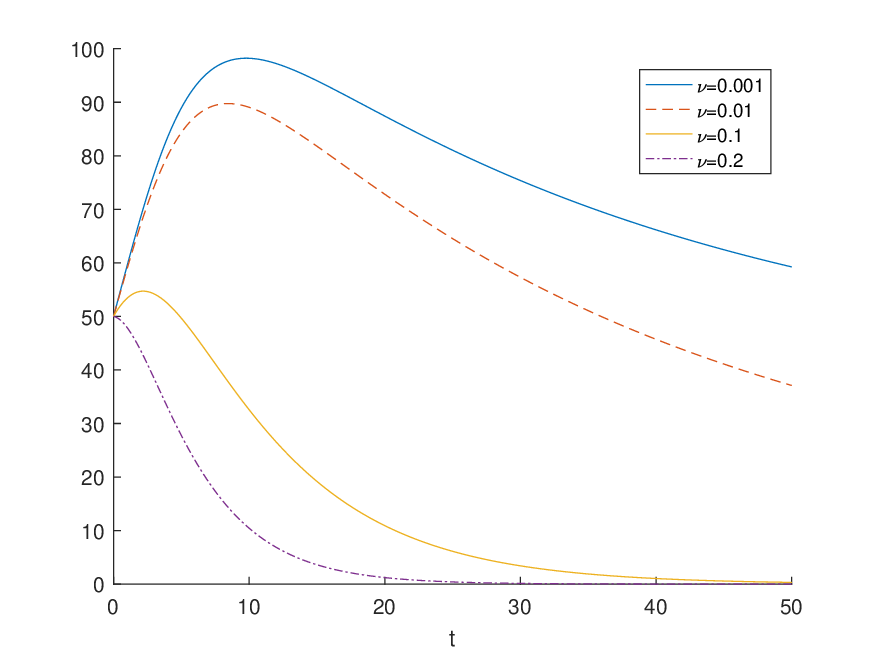}}
       \hfill
       \subfloat[ $\xi_0= 0.0720$]{\includegraphics[width=0.5\textwidth]{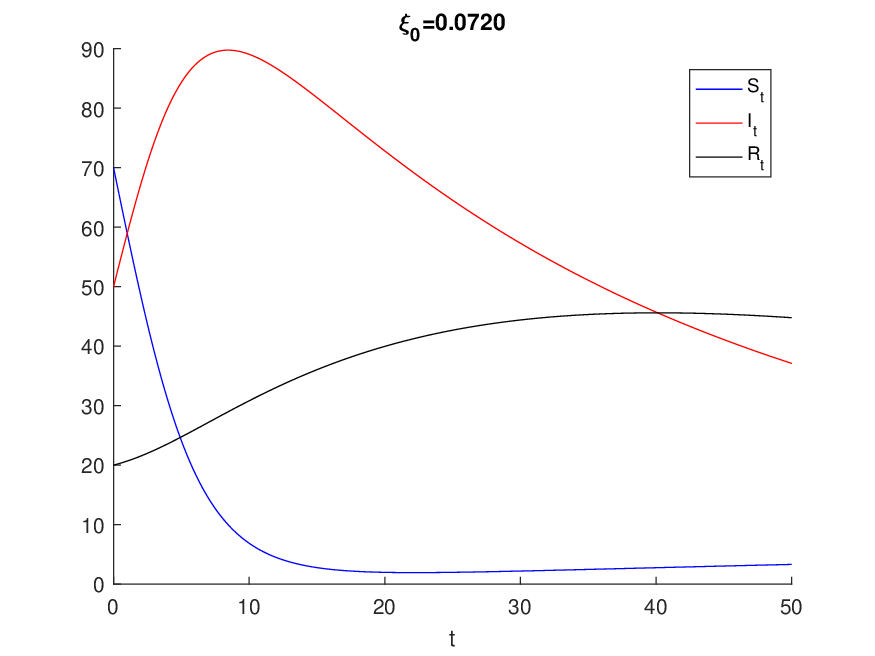}}
 \caption{\label{I_t-Deter} Sample path of $\frac{dI}{dt}$ (a) when $\xi_0=1.0286$ and  $\xi_0=0.0103$. (b) The phaseline of $dI_t/dt$ at different values of $\nu$. (c) When the reproduction number $\xi_0 <1.$}
  \end{figure}

\section{ Dynamics of the stochastic COVID-19 system}\label{Sec3}
\subsection{Existence and uniqueness of the solution}\label{s3}
To study the dynamical behaviour of dynamic biological system,  the main concern is to check whether the solution of the system is unique global and positive. A dynamical system has a  {uniquely}  global solution, if it exhibits no explosion in a given finite time. To have a {uniquely}  global solution, the coefficients of the system must satisfy the following two conditions: (i) local Lipschitz condition, (ii) linear growth condition; see \cite{applebaum2009levy,duan2015introduction}. However, the coefficients of { the} stochastic COVID-19 model (\ref{CV-BM-LM}) do not satisfy the second condition (linear growth condition), so the solution $(S_t,I_t,R_t)$ of system (\ref{CV-BM-LM}) can explode in a finite time $t$. The following Theorem helps us to show that there exists a unique positive solution $(S_t,I_t,R_t)\in \mathbb{R}^3_+$ to COVID-19 system (\ref{CV-BM-LM}).
\begin{theorem}\label{Ext.Soln}
For any given initial condition $(S_0,I_0,R_0)\in \mathbb{R}^3_+$, there is a unique non-negative solution $(S_t,I_t,R_t)\in \mathbb{R}_+^3$ of the model (\ref{CV-BM-LM}) for time $t\geq 0$.
\end{theorem}
\begin{proof}
The differential equation (\ref{CV-BM-LM}) has a locally Lipschitz continuous coefficient, so the model has a unique local solution $(S_t,I_t,R_t)$ on $t\in [0,t_{e})$ where $t_{e}$ is the time for noise for { the} explosion. In order to have a global solution, we need to show that $t_{e}=\infty$ almost surely. To do this, assume that $k_0$ is very large positive number $(k_0>0)$ so that the initial condition $(S_0,I_0,R_0)\in\left[\frac{1}{k_0},k_0\right]$. For every integer $k\geq k_0$, the stopping time is defined as:
\begin{align*}
   \tau_e=inf\{t\in[0,t_{e}):min(S_t,I_t,R_t)\leq \frac{1}{k_0},\,\, or\,\,\, max(S_t,I_t,R_t)\geq k\}{.}
\end{align*}
As $k$ goes to $\infty$, $\tau_k$ increases. Define $lim_{k\rightarrow \infty}\,\tau_k=\tau_{\infty}$ with $\tau_{\infty} \leq \tau_e$. If we can prove that $\tau_{\infty}=\infty$ almost surely, then $\tau_e=\infty.$ If this is false, then there are two positive constants $T>0$ and $\delta\in (0,1)$ such that
\begin{align*}
\mathbb{P}\{\tau_\infty\leq T\} > \delta.
\end{align*}
Thus there is $k_1\geq k_0$ that satisfies
\begin{align*}
\mathbb{P}\{\tau_k\leq T\}\geq \delta, \quad k\geq k_1.
\end{align*}
Now , let us define a $C^2$-function $W$: $\mathbb{R}_+^3 \rightarrow \mathbb{R}_+$ by
\begin{align}\label{C^2-Fun}
W(S,I,R)=(S-\alpha-\alpha\,\frac{\ln S}{\alpha})+(I-1-\ln I)+ (R-1-\ln R){.}
\end{align}
Applying  It$\hat{o}$ formula in Lemma \ref{Ito-form} to Eq. (\ref{C^2-Fun}) yields,
\begin{align}\label{C^2-Fun1}
dW(S,I,R)&=(1-\alpha/S)dS+\frac{(dS)^2}{2\,S^2}+(1-1/I)dI+\frac{(dI)^2}{2\,I^2}+(1-1/S)dS+\frac{(dR)^2}{2\,R^2}\nonumber\\
&:=LW\,dt+ \bar{W},
\end{align}
where $L$ is a differential operator, \cite{applebaum2009levy}.

\begin{align*}
 \bar{W}&=\lambda_1\,SdB^1_t+\int_\mathbb{Y}\epsilon_1(y)S\bar{N}(dt,dy)-\alpha\,\lambda_1\,dB^1_t-\alpha\,\int_\mathbb{Y}\epsilon_1(y)\bar{N}(dt,dy)\nonumber\\
 &+\lambda_2\,IdB^2_t+\int_\mathbb{Y}\epsilon_2(y)I\,\bar{N}(dt,dy)-\lambda_2\,dB^2_t-\int_\mathbb{Y}\epsilon_2(y)\bar{N}(dt,dy)\nonumber\\
 &+\lambda_3\,RdB^3_t+\int_\mathbb{Y}\epsilon_3(y)R\,\bar{N}(dt,dy)-\lambda_3\,dB^3_t-\int_\mathbb{Y}\epsilon_3(y)\bar{N}(dt,dy),
\end{align*}
and

$LW:\mathbb{R}_+^3 \rightarrow \mathbb{R}_+$ is defined as
\begin{align*}
LW&=\Lambda-\nu\,S+\sigma\, R-\alpha\,\frac{\Lambda}{S}+\alpha\,\beta\,I+\alpha\,\nu-\alpha\,\frac{\sigma\,R}{S}+\frac{\lambda_1^2}{2}+\int_{\mathbb{Y}}\epsilon_1^2(y)\pi(dy)-(\nu+\gamma)\,I-\beta\,S+(\nu+\gamma)+ \frac{\lambda_2^2}{2}\nonumber\\
&+\int_{\mathbb{Y}}\epsilon_2^2(y)\pi(dy) +\gamma\,I-(\nu+\sigma)\,R-\gamma+(\nu+\sigma)
+\frac{\lambda_3^2}{2}+\int_{\mathbb{Y}}\epsilon_3^2(y)\pi(dy)\nonumber\\
&\leq \Lambda+\alpha\,\nu+(\alpha\,\beta\,-(\nu+\gamma))\,I+(\nu+\gamma) -\gamma+(\nu+\sigma)+\frac{\lambda_1^2}{2}+ \frac{\lambda_2^2}{2}+\frac{\lambda_3^2}{2}\nonumber\\
&+\int_{\mathbb{Y}}\epsilon_1^2(y)\pi(dy)+\int_{\mathbb{Y}}\epsilon_2^2(y)\pi(dy)
+\int_{\mathbb{Y}}\epsilon_3^2(y)\pi(dy).
\end{align*}
By plugging in  $\alpha=\frac{\nu+\gamma}{\beta}$, we get
\begin{align*}
  LW\leq &\Lambda+\alpha\,\nu+(\nu+\gamma) -\gamma+(\nu+\sigma)+\frac{\lambda_1^2}{2}+ \frac{\lambda_2^2}{2}+\frac{\lambda_3^2}{2}+\int_{\mathbb{Y}}\epsilon_1^2(y)\pi(dy)+\int_{\mathbb{Y}}\epsilon_2^2(y)\pi(dy)
+\int_{\mathbb{Y}}\epsilon_3^2(y)\pi(dy) \nonumber\\
:=& C.
\end{align*}
where the parameter C is a positive constant. The rest of the proof\, follows Cai et al. \cite[ Lemma 2.2]{Cai2017A}, and Zhu et al. \cite[Theorem 1]{2015A}.
\end{proof}

\subsection{Extinction and persistence of the disease}\label{s4}
Since this paper is considering the epidemic dynamic systems, we are focused in prevail and persist of the COVID-19 in a population.
\subsubsection{Extinction of the disease}
In this subsection, we give some conditions for { the} extinction of COVID-19 in the {stochastic COVID-19} system (\ref{CV-BM-LM}). Since {the} extinction of disease (epidemics) in small populations has the major challenges in
population dynamics \cite{2017Epidemic}. So it is important to study the extinction of COVID-19.\\

Define a parameter $\xi$ as
\begin{align*}
 \xi=\frac{\beta\,\Lambda}{\nu}\frac{1}{\gamma+\nu+\varphi_2},
\end{align*}
where $\varphi_2=\frac{1}{2} \lambda_2+\int_{\mathbb{Y}}[\epsilon_2(y)-\ln (1+\epsilon_2(y))]\pi(dy).$ Here $\xi$ is the basic reproduction number for stochastic COVID-19 model (\ref{CV-BM-LM}).
\begin{remark}
From $(e)$ and Remark \ref{rem-1}, we have
\begin{align*}
 \varphi_2&=  \frac{\lambda^2_{2}}{2}+\int_{\mathbb{Y}} [\epsilon_{2}(y)-\ln(1+\epsilon_{2}(y))]\,\pi(dy) \nonumber\\
 &=\frac{\lambda^2_{2}}{2}+\int_{\mathbb{Y}}[(1+\epsilon_{2}(y))-1-\ln(1+\epsilon_{2}(y))]\,\pi(dy) \nonumber\\
 &\geq \frac{\lambda^2_{2}}{2}.
\end{align*}
\end{remark}
\begin{definition}\label{def-1}
For the stochastic model (\ref{CV-BM-LM}) if $lim_{t\rightarrow\infty}I_t=0$, then the disease $I_t$ is said to be extinct,  a.s.\\
\end{definition}
\begin{theorem}\label{Exten-thr}
Assume that $(g)$ holds. Then for any initial condition $(S_0,I_0,R_0)\in \mathbb{R}^3_{+}$, the solution $(S_t,I_t,R_t) \in \mathbb{R}^3_{+}$  of the stochastic COVID-19 model (\ref{CV-BM-LM}) has the following properties:

\begin{align*}
 lim_{t\rightarrow \infty} sup \frac{\ln I_t}{t}\leq \beta\,\frac{\Lambda}{\nu}\left(1-\frac{1}{\xi}\right),\quad a.s.
\end{align*}
If $\xi< 1$ holds, then $I_t$ can go to zero with probability one.\\
Moreover,

$lim_{t\rightarrow \infty}<S>_t=\frac{\Lambda}{\nu}=S_0$,\,\, $lim_{t\rightarrow \infty}<R>_t=0,$ \quad a.s.
\end{theorem}
\begin{proof}
Integrating both sides of model (\ref{CV-BM-LM}) and dividing by $t$, gives

\begin{align}\label{S}
    \frac{S_t-S_0}{t}=\Lambda-\beta\,<S>_t\,<I>_t-\nu<S>_t+\sigma\,<R>_t+\frac{\lambda_1}{t}\int_0^tS_rdB^1_r+\frac{1}{t}\int_0^t\int_{\mathbb{Y}}\epsilon_1(y)\,S_r\,\bar{N}(dr,dt),
\end{align}

\begin{align}\label{I1}
    \frac{I_t-I_0}{t}=\beta\,<S>_t\,<I>_t-(\gamma+\nu)<I>_t+\frac{\lambda_2}{t}\int_0^tI_rdB^2_r+\frac{1}{t}\int_0^t\int_{\mathbb{Y}}\epsilon_2(y)\,I_r\,\bar{N}(dr,dt),
\end{align}

\begin{align}\label{R}
    \frac{R_t-R_0}{t}=\gamma\,<I>_t-(\nu+\sigma)<R>_t+\frac{\lambda_3}{t}\int_0^tR_rdB^3_r+\frac{1}{t}\int_0^t\int_{\mathbb{Y}}\epsilon_3(y)\,R_r\,\bar{N}(dr,dt).
\end{align}
Multiplying both side of Eq. (\ref{R}) by $\frac{\sigma}{\nu+\sigma}$, we have

\begin{align}\label{R1}
    \frac{\sigma}{\nu+\sigma}\frac{R_t-R_0}{t}=\frac{\sigma}{\nu+\sigma}\,\gamma\,<I>_t-\sigma<R>_t+\frac{\sigma}{\nu+\sigma}\frac{\lambda_3}{t}\int_0^tR_rdB^3_r+\frac{\sigma}{\nu+\sigma}\frac{1}{t}\int_0^t\int_{\mathbb{Y}}\epsilon_3(y)\,R_r\,\bar{N}(dr,dt).
\end{align}
Adding Eqs. (\ref{S}), (\ref{I1}), and (\ref{R1}), we obtain
\begin{align}\label{sum}
\frac{S_t-S_0}{t}+ \frac{I_t-I_0}{t}+   \frac{\sigma}{\nu+\sigma}\frac{R_t-R_0}{t}=&\Lambda-\nu<S>_t+\frac{\lambda_1}{t}\int_0^tS_rdB^1_r+\frac{1}{t}\int_0^t\int_{\mathbb{Y}}\epsilon_1(y)\,S_r\,\bar{N}(dr,dt) \nonumber\\
&-(\gamma+\nu)<I>_t+\frac{\lambda_2}{t}\int_0^tI_rdB^2_r+\frac{1}{t}\int_0^t\int_{\mathbb{Y}}\epsilon_2(y)\,I_r\,\bar{N}(dr,dt)\nonumber\\
&\frac{\sigma}{\nu+\sigma}\,\gamma\,<I>_t+\frac{\sigma}{\nu+\sigma}\frac{\lambda_3}{t}\int_0^tR_rdB^3_r+\frac{\sigma}{\nu+\sigma}\frac{1}{t}\int_0^t\int_{\mathbb{Y}}\epsilon_3(y)\,R_r\,\bar{N}(dr,dt)\nonumber\\
&=\Lambda-\nu<S>_t-\left((\gamma+\nu)-\frac{\sigma}{\nu+\sigma}\,\gamma\right)<I>\nonumber\\
&+\frac{\lambda_1}{t}\int_0^tS_rdB^1_r+\frac{1}{t}\int_0^t\int_{\mathbb{Y}}\epsilon_1(y)\,S_r\,\bar{N}(dr,dt)\nonumber\\
&+\frac{\lambda_2}{t}\int_0^tI_rdB^2_r+\frac{1}{t}\int_0^t\int_{\mathbb{Y}}\epsilon_2(y)\,I_r\,\bar{N}(dr,dt)\nonumber\\
&+\frac{\sigma}{\nu+\sigma}\frac{\lambda_3}{t}\int_0^tR_rdB^3_r+\frac{\sigma}{\nu+\sigma}\frac{1}{t}\int_0^t\int_{\mathbb{Y}}\epsilon_3(y)\,R_r\,\bar{N}(dr,dt).
\end{align}
Rewrite Eq. (\ref{sum}) as
\begin{align}\label{<S_t>}
   <S>_t=\frac{\Lambda}{\nu}-\left(\frac{\gamma+\nu+\sigma}{\nu+\sigma}\right)<I>_t+\bar{\Phi}_t,
\end{align}
where
\begin{align*}
\bar{\Phi}_t=&-\frac{1}{\nu}\left( \frac{S_t-S_0}{t}+ \frac{I_t-I_0}{t}+ \frac{\sigma}{\nu+\sigma}\frac{R_t-R_0}{t} \right)+\frac{1}{\nu}\left(\frac{\lambda_1}{t}\int_0^tS_rdB^1_r +\frac{1}{t}\int_0^t\int_{\mathbb{Y}}\epsilon_1(y)\,S_r\,\bar{N}(dr,dt)\right)\nonumber\\
&+\frac{1}{\nu}\left(\frac{\lambda_2}{t}\int_0^tI_rdB^2_r+\frac{1}{t}\int_0^t\int_{\mathbb{Y}}\epsilon_2(y)\,I_r\,\bar{N}(dr,dt)\right)\nonumber\\ &+\frac{1}{\nu}\left(\frac{\sigma}{\nu+\sigma}\frac{\lambda_3}{t}\int_0^tR_rdB^3_r+\frac{\sigma}{\nu+\sigma}\frac{1}{t}\int_0^t\int_{\mathbb{Y}}\epsilon_3(y)\,R_r\,\bar{N}(dr,dt) \right).
\end{align*}
From Lemma \ref{lemma-1} (i-ii),
\begin{align}\label{phi}
  lim_{t\rightarrow \infty} \bar{\Phi}_t=0,\quad a.s.
\end{align}
Therefore, Eq. (\ref{<S_t>}) becomes
\begin{align}\label{<S_t>-1}
   <S>_t=\frac{\Lambda}{\nu}-\left(\frac{\gamma+\nu+\sigma}{\nu+\sigma}\right)<I>_t.
\end{align}
Setting $Z=\ln I_t$ and applying It$\hat{o}$ formula to $Z$ yields,

\begin{align}\label{dln-I_t}
   dZ= d \ln I_t&=\frac{1}{I_t}dI_t-\frac{1}{2I^2_t}[dI_t]^2\nonumber\\
    &=(\beta S_t-(\nu+\gamma)-\varphi_2)dt+\lambda_2 I_tdB^2_t+\int_{\mathbb{Y}}\ln (1+\epsilon_2(y))\bar{N}(dt,dy){.}
\end{align}
Integrating both sides of Eq. (\ref{dln-I_t}) and dividing by $t$, gives
\begin{align}\label{dln-I_t-2}
  \frac{\ln I_t}{t}= \beta< S>_t-(\nu+\gamma)-\varphi_2+\frac{\lambda_2 I_tdB^2_t}{t}+\frac{1}{t}\int_{\mathbb{Y}}\ln (1+\epsilon_2(y))\bar{N}(dt,dy)+\frac{\ln I_0}{t}.
\end{align}
Upon plugging in $<S>_t$ of Eq. (\ref{<S_t>-1}) into Eq. (\ref{dln-I_t-2}), we get
\begin{align}\label{dln-I_t-3}
  \frac{\ln I_t}{t}&= \beta\left(\frac{\Lambda}{\nu}-\left(\frac{\gamma+\nu+\sigma}{\nu+\sigma}\right)<I>_t\right)-(\nu+\gamma)-\varphi_2+\frac{\lambda_2 I_tdB^2_t}{t}+\frac{1}{t}\int_{\mathbb{Y}}\ln (1+\epsilon_2(y))\bar{N}(dt,dy)+\frac{\ln I_0}{t}\nonumber\\
  &=  \beta\,\frac{\Lambda}{\nu}-(\nu+\gamma)-\varphi_2-\beta\left(\frac{\gamma+\nu+\sigma}{\nu+\sigma}\right))<I>_t+\frac{\lambda_2 I_tdB^2_t}{t}+\frac{1}{t}\int_{\mathbb{Y}}\ln (1+\epsilon_2(y))\bar{N}(dt,dy)+\frac{\ln I_0}{t}\nonumber\\
    &= \beta\,\frac{\Lambda}{\nu}-(\nu+\gamma+\varphi_2)-\beta\left(\frac{\gamma+\nu+\sigma}{\nu+\sigma}\right)<I>_t+\frac{\lambda_2 I_tdB^2_t}{t}+\frac{\psi_2(t)}{t}+\frac{\ln I_0}{t}\nonumber\\
  &\leq \beta\,\frac{\Lambda}{\nu}\left(1-\frac{\nu}{\beta\,\Lambda}(\nu+\gamma+\varphi_2)\right)-\beta\left(\frac{\gamma+\nu+\sigma}{\nu+\sigma}\right)<I>_t+\frac{\lambda_2 I_tdB^2_t}{t}+\frac{\psi_2(t)}{t}+\frac{\ln I_0}{t}\nonumber\\
  &\leq \beta\,\frac{\Lambda}{\nu}\left(1-\frac{1}{\xi}\right)-\beta\left(\frac{\gamma+\nu+\sigma}{\nu+\sigma}\, \right)<I>_t+\frac{\lambda_2 I_tdB^2_t}{t}+\frac{\psi_2(t)}{t}+\frac{\ln I_0}{t}\nonumber\\
  &\leq \beta\,\frac{\Lambda}{\nu}\left(1-\frac{1}{\xi}\right)-\beta\left(\frac{\gamma+\nu}{\nu+\sigma}\, \right)<I>_t+\frac{\lambda_2 I_tdB^2_t}{t}+\frac{\psi_2(t)}{t}+\frac{\ln I_0}{t}, \quad since \quad -\frac{\gamma+\nu+\sigma}{\nu+\sigma}< -\frac{\gamma+\nu}{\nu+\sigma}.
\end{align}
From $(f)$ and theorem of large numbers \cite{mao2007stochastic}
\begin{align}\label{psi}
 lim_{t\rightarrow \infty}\frac{\psi_2(t)}{t}=0, \quad a.s{,}
\end{align}
and
\begin{align}\label{B_t}
lim_{t\rightarrow \infty}\frac{B_t}{t}=0   \quad a.s.
\end{align}
 By applying  superior limit ($lim_{t\rightarrow \infty}\,sup $) on both sides of  to Eq. (\ref{dln-I_t-3}), gives
\begin{align}\label{dln-I_t-4}
 lim_{t\rightarrow \infty} sup \frac{\ln I_t}{t} &\leq  \beta\,\frac{\Lambda}{\nu}\left(1-\frac{1}{\xi}\right), a.s.
\end{align}
\begin{figure}[htb!]
\centering
    \subfloat[The susceptibility graph  ]{\includegraphics[width=0.5\textwidth]{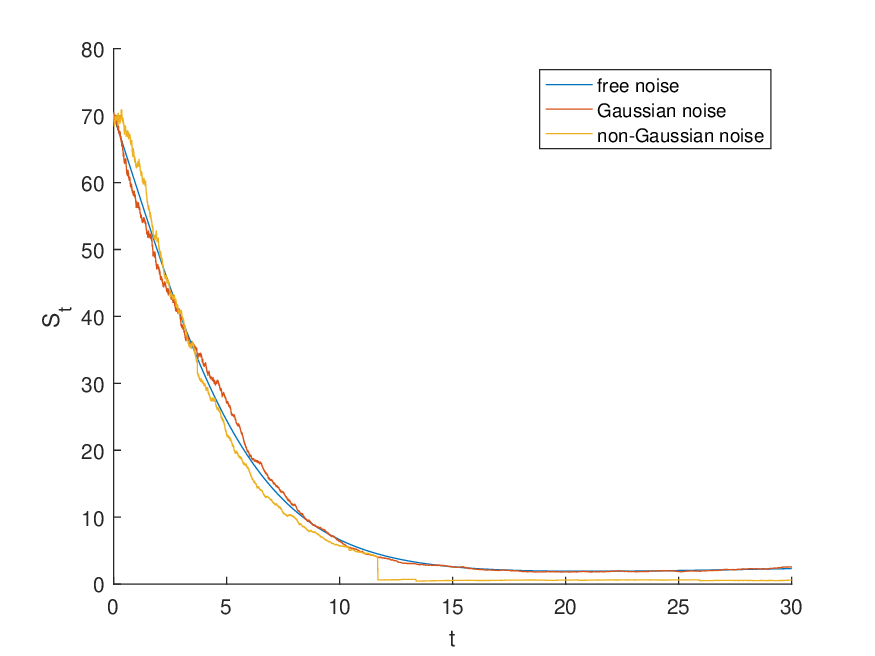}}
\subfloat[The {graph of {the} infected} population by COVID-19]{\includegraphics[width=0.5\textwidth]{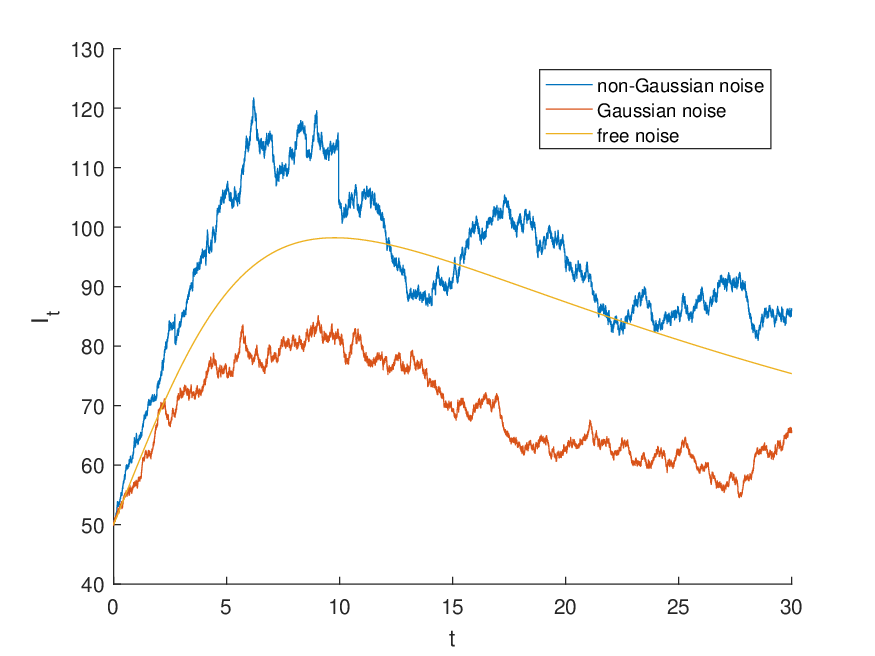}}
    \hfill
\subfloat[The graph of {the} removed people ]{\includegraphics[width=0.5\textwidth]{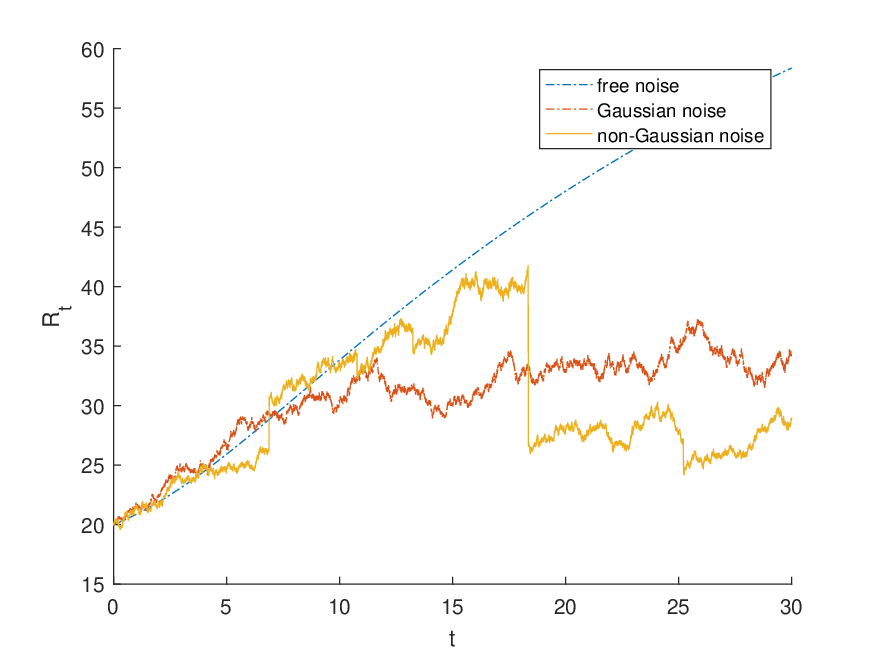}}
    \caption{\label{SIR-1} The numerical results of model (\ref{CV-BM-LM}). (a) The {graph of {the} susceptible}. (b) The {graph of {the} infected} people.(c) The {graph of {the} recoveblack} people. Parameters $S_0=70,\, I_0=50,\, R_0=20,\, \Lambda=0.0072,\,  \beta=0.002,\,\nu=0.001,\,\sigma=0.01,\,\gamma=0.02,\,\lambda_j=0.047$,\,$\epsilon_j(y)=0.004, \,\, j=1,2,3,\,\, \xi=0.9760 <1.$ }
\end{figure}
If $\xi <1$ holds, then $\beta\,\frac{\Lambda}{\nu}\left(1-\frac{1}{\xi}\right)<0$.\\
{ Therefore,}
\begin{align}\label{I}
 lim_{t\rightarrow \infty}I_t=0{.}
\end{align}
 From Definition \ref{def-1}, this implies that $I_t$ can tends to zero with probability one.
Similarly, we can show that
\begin{align}\label{lim-R}
 lim_{t\rightarrow \infty}<R>_t=0{.}
\end{align}
Recall Eq. (\ref{CD_X_sol}),
\begin{align*}
   X=\frac{\Lambda}{\nu}+X_0\,e^{-\nu\,t}.
\end{align*}
Using  Eqs. (\ref{I}) and (\ref{lim-R}), and

\begin{align*}
 lim_{t\rightarrow\infty}X=lim_{t\rightarrow\infty}(S_t+I_t+R_t)= \frac{\Lambda}{\nu},
\end{align*}
we obtain
\begin{align*}
  lim_{t\rightarrow \infty}<S>_t=\frac{\Lambda}{\nu}=S_0{.} \qquad \qquad \qquad \qquad \qquad \qquad \Box
\end{align*}

\end{proof}

\begin{figure}[htb!]
\centering
    \subfloat[The {graph of {the} susceptible}]{\includegraphics[width=0.5\textwidth]{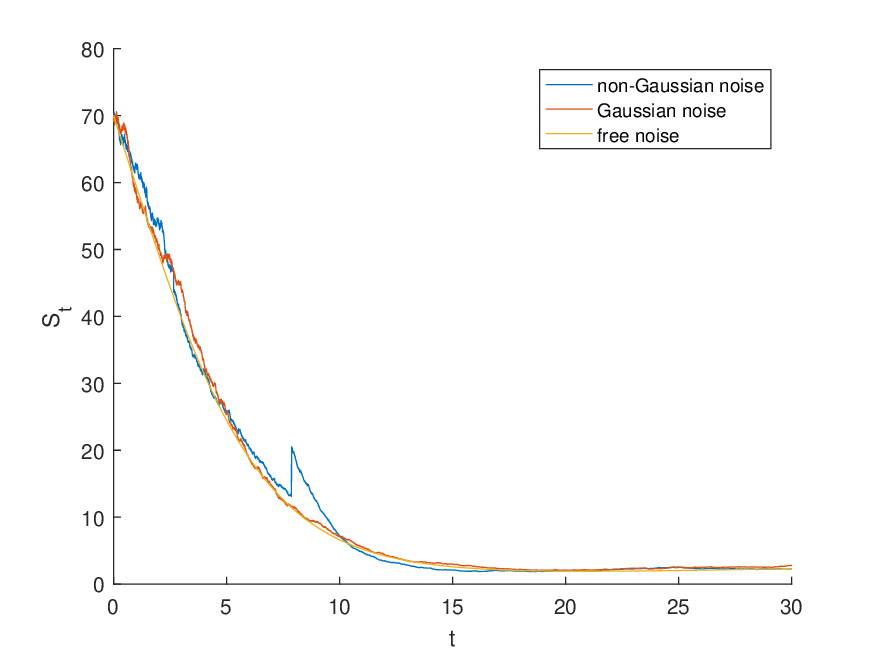}}
    \subfloat[The {graph of {the} infected} population by COVID-19]{\includegraphics[width=0.5\textwidth]{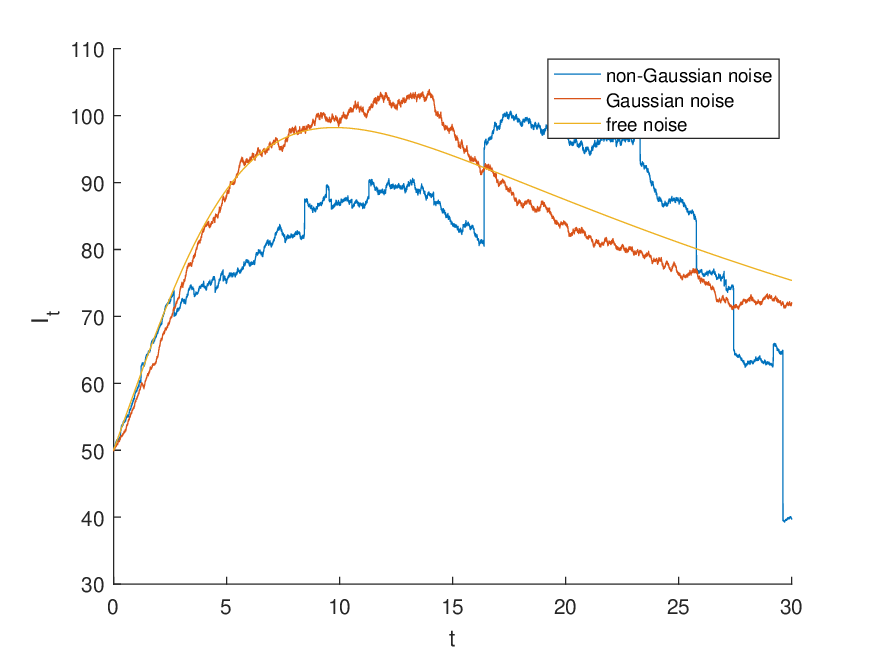}}
    \hfill
        \subfloat[The {graph of {the} recoveblack} people ]{\includegraphics[width=0.5\textwidth]{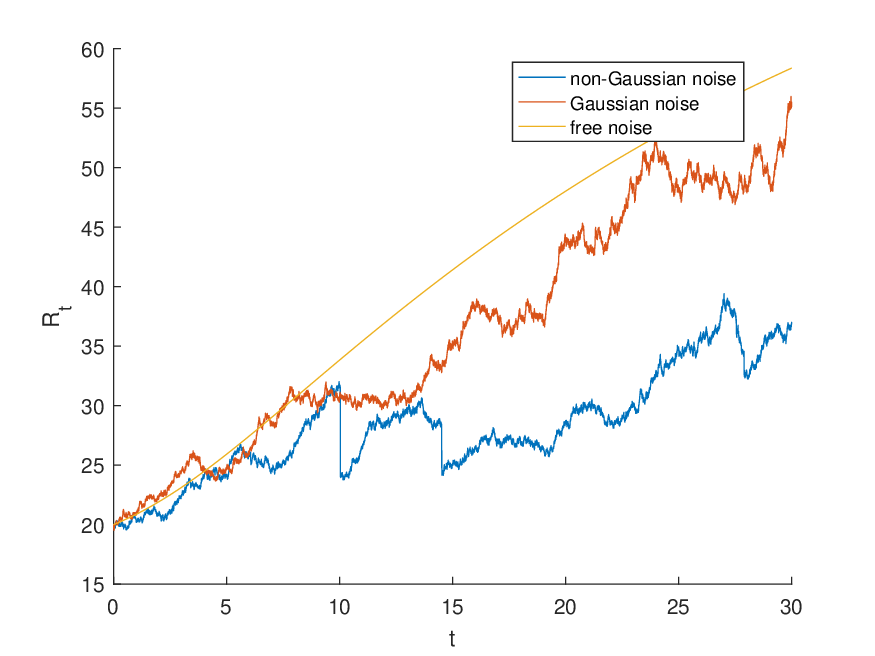}}
    \caption{\label{SIR-2} The numerical simulation of model (\ref{CV-BM-LM}). (a) The {graph of {the} susceptible}. (b) The {graph of the infected} people.(c) The {graph of the recoveblack} people from COVID-19. Parameters $S_0=70,\, I_0=50,\, R_0=20,\, \Lambda=0.0072,\,  \beta=0.002,\,\nu=0.001,\,\sigma=0.01,\,\gamma=0.02,\,\lambda_1=0.047,\, \lambda_2=0.019,\,\, \lambda_3=0.047,\,\epsilon_j(y)=0.004, \, j=1,2,3,\,\, \xi=1.02 > 1$.  }

\end{figure}
\begin{figure}[htb!]
\centering
    \subfloat[Noise free, i.e., $\lambda_j=0$,\,\,$\epsilon_j(y)=0$ ,\,\,$\xi_0=1.0286$]{\includegraphics[width=0.5\textwidth]{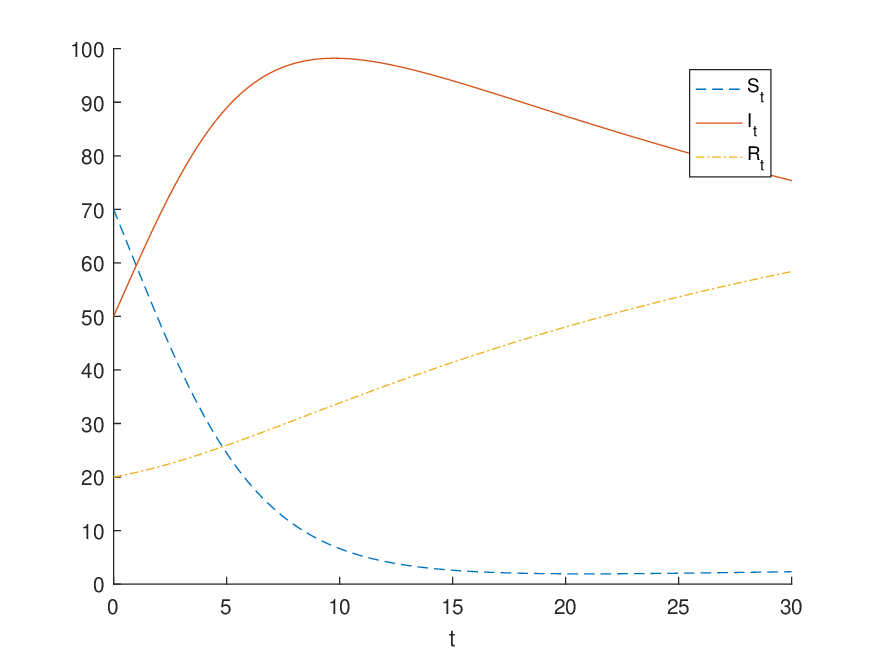}}
    \subfloat[  $\lambda_j=0.047$,\,$\epsilon_j(y)=0,$\,\,$\xi_0=1.0286$.]{\includegraphics[width=0.5\textwidth]{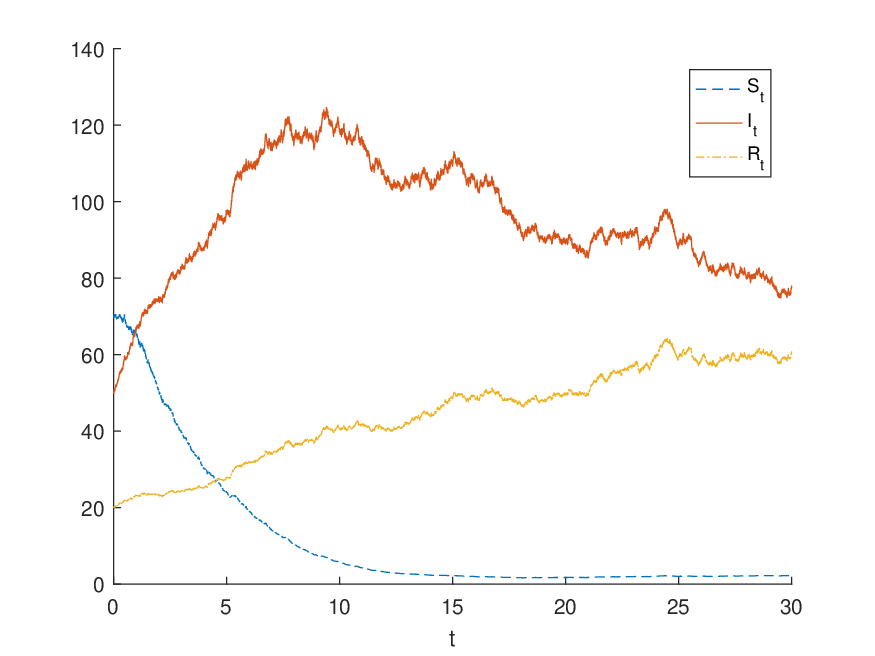}}
    \hfill
    \subfloat[$\lambda_j=0.047$,\,\,$\epsilon_j(y)=0.004$,\,\,$\xi_0=1.0286$.]{\includegraphics[width=0.5\textwidth]{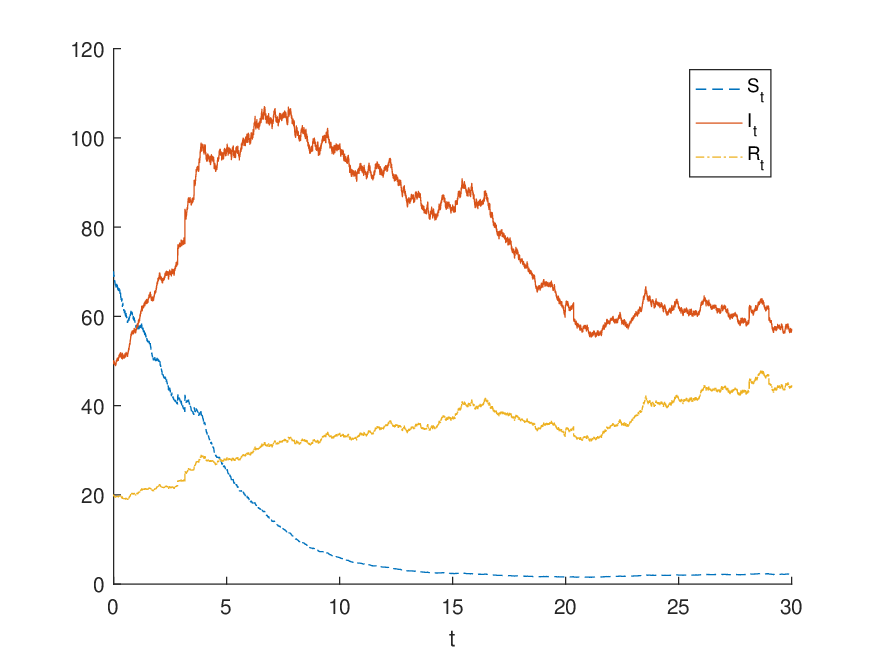}}
    \subfloat[$\lambda_j=0.47$,\,\,$\epsilon_j(y)=0.04$, $\xi_0=0.0720.$]{\includegraphics[width=0.5\textwidth]{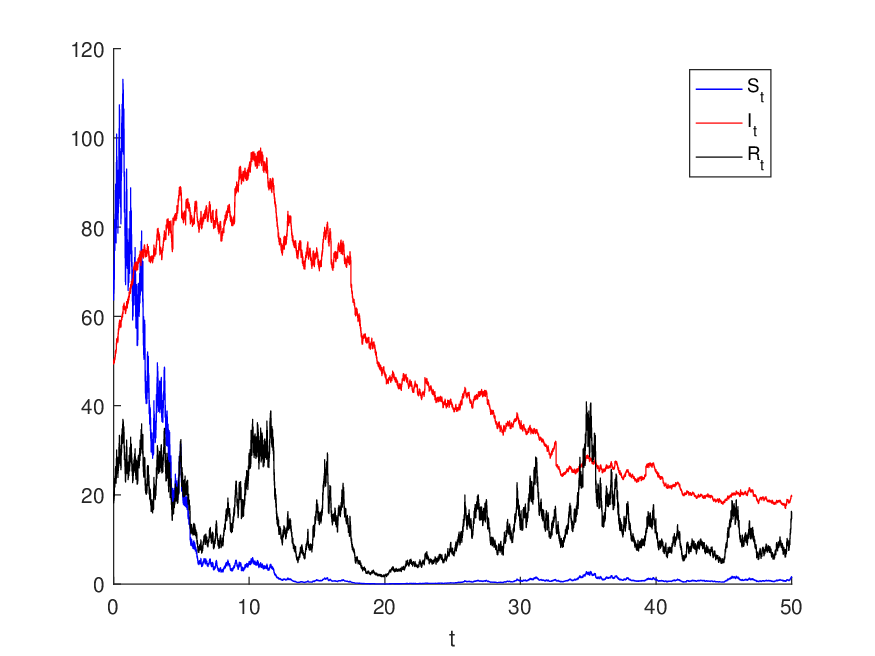}}
    \caption{\label{SIR-3}This Figure shows the numerical simulation of the stochastic COVID-19 model (\ref{CV-BM-LM}) with $S_0=70,\, I_0=50,\, R_0=20,\, \Lambda=0.0072,\,  \beta=0.002,\,\nu=0.001,\,\sigma=0.01,\,\gamma=0.02,\,\, \xi=0.9284 <1$, \, $j=1,2,3$.}

\end{figure}
\subsubsection{Persistence of the disease}
This section deals with the persistence in mean  of the disease in the model (\ref{CV-BM-LM}). Before we state the theorem, {we define} persistence in mean.
\begin{definition}
If $lim_{t\rightarrow \infty}<S>_t\,>0$,\quad  $lim_{t\rightarrow \infty}<I>_t\,>0$,\quad  $lim_{t\rightarrow \infty}<R>_t\,>0$, almost surely, then we can say system (\ref{CV-BM-LM}) is persistence in mean.
\end{definition}
\begin{theorem}\label{per-mean}
For given initial values $(S_0,I_0,R_0)\in \mathbb{R}_{+}^3$, the solution $(S_t,I_t,R_t)\in \mathbb{R}_{+}^3$ of  model (\ref{CV-BM-LM}) exists when  $\xi>1$, \\

$lim_{t\rightarrow \infty}<S>_t=\tilde{S}$,\quad $lim_{t\rightarrow \infty}<I>_t=\tilde{I}$, \quad $lim_{t\rightarrow \infty}<R>_t=\tilde{R}$, a.s,\\
where\\

$\tilde{S}=\frac{\Lambda}{\nu}-\frac{\gamma+\nu+\sigma}{\gamma+\nu}\,\frac{\Lambda}{\nu}\left(1-\frac{1}{\xi}\right)$,\quad $\tilde{I}=\frac{\nu+\sigma}{\gamma+\nu} \frac{\Lambda}{\nu}\left(1-\frac{1}{\xi}\right)$,\quad $\tilde{R}=\frac{1}{\gamma+\nu}\,\frac{\Lambda}{\nu}\left(1-\frac{1}{\xi}\right).$
\end{theorem}
\begin{proof}
Recall Eq. (\ref{dln-I_t-3})
\begin{align}\label{DlnI}
 \frac{\ln I_t}{t}= \beta\,\frac{\Lambda}{\nu}\left(1-\frac{1}{\xi}\right)-\beta\left(\frac{\gamma+\nu}{\nu+\sigma}\, \right)<I>_t+\frac{\lambda_2 I_tdB^2_t}{t}+\frac{\psi_2(t)}{t}+\frac{\ln I_0}{t}{,}
 \end{align}
or equivalently,
\begin{align}\label{<I>-1}
\beta\left(\frac{\gamma+\nu}{\nu+\sigma} \right)<I>_t= -\frac{\ln I_t}{t}+\beta\,\frac{\Lambda}{\nu}\left(1-\frac{1}{\xi}\right)+\frac{\lambda_2}{t} I_tdB^2_t+\psi_2(t)+\frac{\ln I_0}{t}.
\end{align}
From Lemma \ref{lemma-1} and Eqs. (\ref{phi}), (\ref{psi}) and (\ref{B_t}), we get
\begin{align}\label{I-bar}
  lim_{t\rightarrow \infty}  <I>_t=\frac{\nu+\sigma}{\gamma+\nu} \,\frac{\Lambda}{\nu}\left(1-\frac{1}{\xi}\right)=\tilde{I}, \quad a.s.
\end{align}
Substituting Eq. (\ref{I-bar}) into Eq. (\ref{<S_t>-1}), and taking limit on both sides, yields

\begin{equation}\label{S-bar}
lim_{t\rightarrow \infty} <S>_t=\frac{\Lambda}{\nu}-\frac{\gamma+\nu+\sigma}{\gamma+\nu}\,\frac{\Lambda}{\nu}\left(1-\frac{1}{\xi}\right)=\tilde{S}.
\end{equation}
Furthermore, applying $lim_{t\rightarrow\infty}$ to Eq. (\ref{R}) and replacing $<I>_t$ by Eq. (\ref{I-bar}), yields
\begin{align}\label{R-bar}
   lim_{t\rightarrow\infty}<R>_t=\frac{1}{\gamma+\nu}\,\frac{\Lambda}{\nu}\left(1-\frac{1}{\xi}\right)=\tilde{R}.
\end{align}
\end{proof}
The proof is complete. $\Box$
\begin{remark}
From the Theorems \ref{Exten-thr} and \ref{per-mean}  above, we can take the value of $\xi$ as the threshold of the system (\ref{CV-BM-LM}). The value of $\xi$ indicates the prevalence and extinction of the COVID-19. Here, we can observe that the Gaussian and jump noises have {a} significant effect on the {behavior} of the system (\ref{CV-BM-LM}).
\end{remark}

\section{Discussion and numerical experiments}\label{s6}
This section deals with the theoretical results of the investigated deterministic and stochastic epidemic systems by applying numerical simulations. Here to find out the impact of Gaussian  and non-Gaussian noises intensities on this epidemic dynamics, we compare the trajectories of the  deterministic and stochastic  systems. {We choose the initial value $(S_0,I_0,R_0)=(70,50,20),\,\,\Lambda=0.0072,\, \, \beta=0.002,\,\,\sigma=0.01,\,\, \,\,\gamma=0.02$. The other values of the parameters are given in the figures.}

{ Figure \ref{I_t-Deter} plots the numerical simulation of the deterministic epidemic model (\ref{CV-Det})}. Fig. \ref{I_t-Deter}a shows the results of Theorem \ref{Det-thrm} for different {values} of the reproductive number $\xi_0$. We can easily see from the results, the infectious disease of system (\ref{CV-Det}) goes to extinction for $\xi_0 <1$, almost surely, whereas the disease persists if $\xi_0 >1.$ {The parameter $\nu$ will lead to a decrease in $\xi_0$. This tells us the extinction of the disease is very fast as the $\nu$ increases, this phenomena is plotted in Fig.\ref{I_t-Deter}b}. As $\nu$ increases the value of $\xi_0$ is less than one, thus according to Theorem \ref{Det-thrm}, asymptotically results into
extinction of the COVID-19 in the population i.e., $I_t$ can go to zero with probability one. The phase line of COVID-19 epidemic model (\ref{CV-Det}) is given in Fig. \ref{I_t-Deter}c when $\xi_0 <1$ and $\nu=0.01$.

In Figures \ref{SIR-1} and \ref{SIR-2}, we fixed the  parameters $ \nu=0.001$, $\epsilon_j(y)=0.004,$ for $j=1,2,3,$ and $\mathbb{Y}=(0,\infty),\, \pi(\mathbb{Y})=1.$ Here, the value of the basic reproductive number $\xi_0$ is $1.0286$, and $\xi=0.9349$. Having these values, the solution $(S_t,I_t,R_t)$ of the system (\ref{CV-BM-LM}) satisfies the property in Theorem \ref{Exten-thr}, i.e.,
\begin{align*}
    lim_{t\rightarrow \infty}\frac{\ln I_t}{t}\leq \beta\,\frac{\Lambda}{\nu}\left(1-\frac{1}{\xi}\right) =-0.0015 < 0 \quad a.s.
\end{align*}
{This shows} the $I_t$ can vanish as $t$ goes to infinity. This happens because of the L\'evy noise effect. When $\lambda_2=0.019$ and $\xi=1.0093$, the solution $(S_t,I_t,R_t)$ of  Model (\ref{CV-BM-LM}) satisfies the condition in Theorem  \ref{per-mean}.
This scenario means that
\begin{align*}
   lim_{t\rightarrow \infty} <S>_t=7.1025,
\end{align*}
\begin{align*}
  lim_{t\rightarrow \infty}  <I>_t=0.0346,
\end{align*}
and
\begin{align*}
   lim_{t\rightarrow\infty}<R>_t=3.1460, \quad a.s.
\end{align*}
This numerical experiment shows that the COVID-19 will prevail. Noting that, Fig. \ref{SIR-1} and Fig. \ref{SIR-2} only differ by the value of $\lambda_2.$ The relationship of the variables $S_t,\,I_t,\,$ and $R_t$ is plotted in Figure \ref{SIR-3}. {When the reproductive number is $\xi_0$ less than 1, the stochastic reproductive number $\xi$ is also less than 1. For this case, the sample paths of the stochastic COVID-19 model is plotted in Figs.\ref{SIR-3}b, \ref{SIR-3}c and \ref{SIR-3}d. }

The numerical solutions imply that blackucing contact rate, washing hands, improving treatment rate, and environmental sanitation are the most crucial activities to eradicate COVID-19 disease from the community.
\section{{Conclusion}}\label{con}
The non-Gaussian noise plays a significant role in evolution of the epidemic dynamical
processes, like HIV, SARS, avian influenza, and so on. In this work, we have studied the stochastic COVID-19 epidemic model driven by both Gaussian and non-Gaussian noises. In Theorem \ref{Ext.Soln} , we proved that {the model} (\ref{CV-BM-LM}) has a unique non-negative solution. We also investigated some conditions for the extinction and persistence {during} the COVID-19 epidemic. We have applied a matlab  programm to study the behavior of the solution of the model. We have illustrated from numerical results the changing impact of the
noise intensities and the parameter $\nu$ on the size of infectious individuals. The results established in the present study can be used to examine dynamical behaviors for COVID-19, HIV, SARS, and so on.

By using the Euler Maruyama (EM) method \cite{higham2001algorithmic,1992Higher}, we gave some numerical {solutions} to illustrate the extinction and persistence of the disease in {the} deterministic system and stochastic {counterparts} for comparison.  We also obtained and compablack the basic reproduction numbers  for the deterministic model as well as  the stochastic model. From the comparison, we observed that the basic reproduction number of the stochastic COVID-19  model is much smaller than that of the deterministic COVID-19 model, {this} shows that the stochastic approach is more realistic than {the} deterministic one. In other words, the jump noise and white noise can change the behaviour of the model. The noises can force COVID-19 ( disease ) {to go out} extinct.

Furthermore, we showed that the disease can go to extinction if $\xi <1.$ While the COVID-19 becomes persistent for $\xi >1$; see Theorems \ref{Exten-thr} and \ref{per-mean}.

 From the findings, we concluded that if $\xi <1$, it is possible that the spread of the disease can be controlled, but for $\xi >1$, COVID-19 can be persistent. $\frac{\beta\,\Lambda}{\nu}\geq \varphi_2$ implies that the Gaussian and non-Gaussian noises are small. From this result, we conclude that efforts should be encouraged in order to achieve a disease-free population.


\section*{Availability of data and materials}
 The authors confirm that the data supporting the findings of this study are available within the articles cited { therein}.

\section*{Acknowledgment}
This research was supported by King Abdulaziz University Jeddah Saudi Arabia and partially supported by the NSFC grants {12001213}.
%

\section*{Competing interests}
 The authors declare that there is no conflict of interest regarding the publication of this paper.


\section*{Authors Contribution}
Authors have equally contribution in preparing this manuscript.


\section*{Funding}
Supported by King Abdulaziz University Jeddah KSA.

\fontsize{10}{10}\selectfont{\bibliography{References.bib}}
\end{document}